\documentclass[aps]{revtex4}
 \usepackage{amssymb} \usepackage{graphicx}
\usepackage[utf8]{inputenc}
\usepackage[english]{babel} 
\usepackage{amsthm}
\usepackage{tikz}

\newtheorem{proposition}{Proposition}

\begin{document}
 \title{Weyl semimetals and spin$^c$ cobordism}

\author{\"Umit Ertem}
 \email{umitertemm@gmail.com}
\address{Independent Researcher, Ankara, Turkey\\}

\begin{abstract}

Classification of topological insulators and superconductors is manifested in terms of spin cobordism groups for lower dimensions. It is discussed that the periodic table of topological insulators is a result of the possible choices of spin structures on Brillouin zones of relevant topological materials. This framework is extended to the case of Weyl semimetals. It is shown that the classification of Weyl semimetals can be managed in terms of spin$^c$ cobordism groups via the extension of spin structures to spin$^c$ structures. Topological invariants of Weyl semimetals are connected to the topological invariants of spin$^c$ cobordism groups and Fermi arcs of Weyl semimetals are interpreted in terms of the choices of spin$^c$ structures.

\end{abstract}

\maketitle

\section{Introduction}

The classification of the various types of recently discovered topological phases of matter is one of the prominent research topics in theoretical condensed matter physics literature. Topological phases of matter manifest themselves as the topologically protected band gaps or band touching points in the band structure of different classes of Hamiltonians. Topological insulators and topological superconductors have topologically protected bulk band gaps and edge conducting states which differ them from the ordinary insulators \cite{Hasan Kane, Qi Zhang}. Low energy excitations of the band structure of topological insulators can be described by Dirac Hamiltonians. Different types of anti-unitary symmetries, like time-reversal and charge conjugation, give way to the existence of topological insulators and they can be characterized by topological invariants defined from the eigenstates of the corresponding Dirac Hamiltonians. There are two types of topological insulators and superconductors that are $\mathbb{Z}$-insulators and $\mathbb{Z}_2$-insulators  and the topological invariants corresponding to these classes take values in $\mathbb{Z}$ and $\mathbb{Z}_2$ groups which are described by Chern numbers or winding numbers and Kane-Mele invariants or Chern-Simons invariants, respectively \cite{Ryu Schnyder Furusaki Ludwig, Wang Qi Zhang}. In fact, all these topological invariants correspond to the indices of Dirac operators defined from the Dirac Hamiltonians of topological materials \cite{Ertem}. Classification topological insulators and superconductors can be maintained by using the Bott periodicity of K-theory groups and Clifford algebras which results to a periodic table of topological insulators and superconductors for different symmetry classes and dimensions \cite{Kitaev, Stone Chiu Roy}. The periodic table of topological insulators and superconductors can also be obtained from the Clifford chessboard of Clifford algebras and the index theorems of Dirac operators \cite{Ertem}.

Another type of topological materials is topological semimetals which have topologically protected band touching points or nodal lines in their band structure. Low energy excitations of topological semimetals are also described by Dirac Hamiltonians. Those which have topologically protected band touching points in the band structure are called Weyl semimetals or Dirac semimetals depending on the symmetry structure of the Hamiltonians. On the other hand, topological semimetals having band touching nodel lines are called topological nodal-line semimetals \cite{Bernevig Weng Fang Dai, Turner Vishwanath}. For Weyl semimetals, there is no symmetry restriction on the corresponding Hamiltonians and adding mass terms or perturbative terms cannot open a gap in the band structure which only moves the band touching points through the Brillouin zone (BZ). Topologically protected band touching points in Weyl semimetals are called Weyl points. Weyl semimetals correspond to the boundary phases of topological insulators since there must be a gapless phase between the boundary of different classes of topological insulators which correspond to a Weyl semimetal phase. In this respect, Weyl semimetals can be described as the extensions of topological insulators to gapless semimetal phases. Topological invariants characterizing the Weyl semimetals arise from the Chern numbers calculated over the spheres around the Weyl points. Moreover, depending on the symmetry structure of the Hamiltonian, band structure out of the Weyl points can give a nontrivial topological invariant equal to the corresponding topological insulator phase \cite{Mathai Thiang1, Thiang Sato Gomi}. Since Weyl points can have different Chern numbers, there will be different types of gapped phases between different Weyl points and these gapped phases will have different edge structures at the boundary. For nontrivial gapped phase parts between Weyl points, there will be nontrivial edge currents which correspond to the Fermi arcs in Weyl semimetals which are the experimental signatures of these materials \cite{Xu et al, Mathai Thiang2}. Although the topological invariants characterizing Weyl semimetals can be described for various cases, a complete classification scheme for different types and dimensions as in the topological insulators case is not yet known.

In this paper, we consider the classification of topological insulators and superconductors in terms of K-theory groups and relate it to spin cobordism groups. Spin cobordism groups classify the possible choices of spin structures on manifolds and we show that the periodic table of topological insulators and superconductors can be derived from the possible choices of spin structures on the BZ of topological materials for dimensions less than or equal to seven. Hence, topological invariants classifying the spin cobordism classes are equivalent to the topological invariants corresponding to relevant topological insulators. Moreover, spin structures can be extended to spin$^c$ structures via characteristic surfaces depending on some topological restrictions. We discuss that the extension of spin structures to spin$^c$ structures is equivalent to the extension of topological insulators to Weyl semimetals via the spheres around Weyl points. From this equivalence, we show that Weyl semimetals can be characterized by the possible choices of spin$^c$ structures on the BZ and hence they can be classified via spin$^c$ cobordism groups. Topological invariants corresponding to the topological insulator parts of Weyl semimetals are determined by the indices of spin$^c$ Dirac operators and the equations determining the monopole charges around Weyl points are in the same form with Seiberg-Witten (SW) equations which determine the possible choices of spin$^c$ structures. So, topological invariants of Weyl semimetals are equivalent to the topological invariants of spin$^c$ cobordism classes. On the other hand, it is also discussed that the Fermi arcs on the surfaces of Weyl semimetals are direct consequences of the choices of spin$^c$ structures  and hence the spin$^c$ cobordism classes of Weyl semimetals.

The paper is organized as follows. In Section 2, we discuss the general properties of topological insulators and Weyl semimetals. Section 3 deals with the classification of topological insulators and superconductors in terms of spin cobordism groups and corresponding topological invariants. In Section 4, the classification of Weyl semimetals in terms of spin$^c$ cobordism groups is considered and correspondences of topological invariants and Fermi arcs with spin$^c$ structures are discussed. Section 5 concludes the paper.

\section{Topological materials}

Low energy effective Hamiltonians of topological insulators and superconductors are characterized by gapped Dirac Hamiltonians in various dimensions and symmetry classes. In terms of the momentum variables $\textbf{k}$, these Hamiltonians are given in the following form
\begin{equation}
H(\bf{k})=\bf{d}(\bf{k}).\bf{\sigma}
\end{equation}
where $\bf{d}(\bf{k})$ are functions of $\textbf{k}$ and $\bf{\sigma}$ are Clifford algebra generators. The eigenvalue equation of this Hamiltonian can be written in terms of $n$ eigenstates $u_n(\bf{k})$ and $n$ eigenvalues $E_n(\bf{k})$ as
\begin{equation}
H({\bf{k}})u_n({\bf{k}})=E_n({\bf{k}})u_n({\bf{k}}).
\end{equation}
$u_n(\bf{k})$ correspond to the sections of the Bloch bundle over the BZ of the system and the Dirac Hamiltonian given in (1) correspond to the Dirac operator on this bundle. The characteristic property of the Dirac Hamiltonians of topological insulators is that they are stable against perturbations, namely the band gaps of these Hamiltonians cannot be closed by symmetry preserving perturbations without changing the topological class of the system. To convert one topological class Hamiltonian to another one, there must be a gapless phase between these two classes. Hence, the insulating phase must disappear at the boundary and these correspond to edge currents of the system at the boundary. Topological phases of topological materials are characterized by topological invariants defined from the eigenvalues of the Dirac Hamiltonians. There are two types of topological invariants for topological materials depending on the symmetry class and dimension of the system. These are called $\mathbb{Z}$ and $\mathbb{Z}_2$ invariants depending on the group that the topological invariants take values. For example, $\mathbb{Z}$ invariant can be constructed from the eigenstates of the Dirac Hamiltonian by defining the Berry connection
\begin{equation}
A_n({\bf{k}})=i<u_n({\bf{k}})|\nabla_k|u_n({\bf{k}})>
\end{equation}
and Berry curvature $F_n({\bf{k}})=\nabla_k\times A_n({\bf{k}})$. Then, the topological invariant correspond the integral of the Berry curvature over the BZ which is the first Chern number of the Bloch bundle
\begin{equation}
C_1=\frac{1}{2\pi}\int_{BZ}F({\bf{k}})dk.
\end{equation}
Similarly, for a $\mathbb{Z}_2$ class topological insulator, the topological invariant can be constructed from the eigenstates by defining the sewing matrix
\begin{equation}
w_{ij}({\bf{k}})=<u_i(-{\bf{k}})|T|u_j({\bf{k}})>
\end{equation}
where $T$ is the time-reversal operator. Then, the Kane-Mele $\mathbb{Z}_2$ invariant is defined along the time-reversal invariant momentum points $\Lambda_k$ in terms of the Pfaffian and determinant of the sewing matrix as follows
\begin{equation}
(-1)^{\nu}=\prod_{\lambda\in\Lambda_k}\frac{\text{Pf}(w(\lambda))}{\sqrt{\det w(\lambda)}}.
\end{equation}
In fact, these topological invariants correspond to the indices of Dirac operators corresponding to the Dirac Hamiltonians of topological materials \cite{Ertem}.

There also exist topological materials that have no band gaps which are called topological semimetals. The defining characteristic of topological semimetals is that they have band touching points or nodes at the Fermi energy of the band structure. In ordinary metals these band touching points can be removed by perturbations. However, in topological semimetals the band touching points are stable against the symmetry preserving perturbations. The prominent examples of topological semimetals are Weyl semimetals whose low energy excitations around band touching points satisfy the Weyl equation which corresponds to the massless Dirac quation for chiral spinors. So, the band touching points in Weyl semimetals are called Weyl points. The Bloch Hamiltonian for Weyl semimetals is given in a similar way to other Dirac materials
\begin{equation}
H({\bf{k}})={\bf{d}}({\bf{k}}).{\bf{\sigma}}
\end{equation}
and it has zero eigenvalues at Weyl points. Spectrum of the Hamiltonian corresponds to $\pm|{\bf{h}}({\bf{k}})|$ and for ${\bf{h}}({\bf{k}})=0$, we have the Weyl points. When we add mass terms or other perturbation terms to the Hamiltonian, the Weyl points will move along the BZ but the Hamiltonian cannot be gapped in this way. However, if there are two Weyl points with opposite charges, then they can move towards each other and annihilate each other which results a gap in the band structure. Weyl semimetals can be thought as the boundary phases between different topological insulator classes since they correspond to band touching phases between different gapped topological classes. Weyl semimetals can be characterized by two types of topological invariants. If we omit the Weyl points in the band structure, then the resulting band will correspond to a gapped topological insulator and the topological invariants characterizing this topological insulator phase will also characterize the gapped part of the Weyl semimetal. The second topological invariant characterizing Weyl semimetals comes from the monopole charges calculated around the Weyl points. If we draw a 2-dimensional sphere $S_w$ around the Weyl point $w$, we can calculate a Chern number around this surface from the defined Berry curvature of the Hamiltonian which is in the form
\begin{equation}
F_{ij}=\epsilon_{ijl}\frac{{\bf{k}}_l}{2|{\bf{k}}|^3}
\end{equation}
where $\epsilon_{ijl}$ is the antisymmetric tensor. In fact, this invariant corresponds to the index of the restricted Hamiltonian $\hat{{\bf{h}}}_w$ of ${\bf{h}}({\bf{k}})$ to $S_w$. So, the topological invariants of gapped parts and the monopole charges around Weyl points correspond to the topological invariants of Weyl semimetals. One of the most interesting features of the existence of Weyl points in Weyl semimetals is the presence of Fermi arcs. For example, if there are two Weyl points in a semimetal, then there will be different topological invariants between and out of Weyl points which correspond to the gapped part invariants. So, there is a phase change between gapped part invariants and if we consider the projection of Weyl points to the boundary, then there will be a local edge current between Weyl points when there is a nontrivial topological invariant in the bulk between them. These local edge currents at the boundary are called Fermi arcs and they are consequences of the gapped part invariants of Weyl semimetals.

\section{Topological insulators and spin cobordism groups}

Bulk insulating and edge conducting topological materials can be classified in an organized manner for various dimensions and symmetry classes. For different symmetry classes corresponding to Hamiltonians in the presence or absence of time-reversal (T), charge conjugation (C) and chiral (S) symmetries, a periodic table determining the presence or absence of topological phases in various dimensions can be constructed as in Table I \cite{Kitaev}. Topological phases in $\mathbb{Z}$ or $\mathbb{Z}_2$ classes are classified by relevant topological invariants related to the indices of the Dirac Hamiltonians \cite{Ertem}. 

\begin{table}[h]
\centering
\begin{tabular}{c| c c c|c c c c c c c c}

\hline \hline
$\text{label}$ & $T$ & $C$ & $S$ & $0$ & $1$ & $2$ & $3$ & $4$ & $5$ & $6$ & $7$ \\ \hline
$\text{A}$ & $0$ & $0$ & $0$ & $\mathbb{Z}$ & $0$ & $\mathbb{Z}$ & $0$ & $\mathbb{Z}$ & $0$ & $\mathbb{Z}$ & $0$ \\
$\text{AIII}$ & $0$ & $0$ & $1$ & $0$ & $\mathbb{Z}$ & $0$ & $\mathbb{Z}$ & $0$ & $\mathbb{Z}$ & $0$ & $\mathbb{Z}$ \\ \hline
$\text{AI}$ & $1$ & $0$ & $0$ & $\mathbb{Z}$ & $0$ & $0$ & $0$ & $\mathbb{Z}$ & $0$ & $\mathbb{Z}_2$ & $\mathbb{Z}_2$ \\
$\text{BDI}$ & $1$ & $1$ & $1$ & $\mathbb{Z}_2$ & $\mathbb{Z}$ & $0$ & $0$ & $0$ & $\mathbb{Z}$ & $0$ & $\mathbb{Z}_2$ \\
$\text{D}$ & $0$ & $1$ & $0$ & $\mathbb{Z}_2$ & $\mathbb{Z}_2$ & $\mathbb{Z}$ & $0$ & $0$ & $0$ & $\mathbb{Z}$ & $0$ \\
$\text{DIII}$ & $-1$ & $1$ & $1$ & $0$ & $\mathbb{Z}_2$ & $\mathbb{Z}_2$ & $\mathbb{Z}$ & $0$ & $0$ & $0$ & $\mathbb{Z}$ \\
$\text{AII}$ & $-1$ & $0$ & $0$ & $\mathbb{Z}$ & $0$ & $\mathbb{Z}_2$ & $\mathbb{Z}_2$ & $\mathbb{Z}$ & $0$ & $0$ & $0$ \\
$\text{CII}$ & $-1$ & $-1$ & $1$ & $0$ & $\mathbb{Z}$ & $0$ & $\mathbb{Z}_2$ & $\mathbb{Z}_2$ & $\mathbb{Z}$ & $0$ & $0$ \\
$\text{C}$ & $0$ & $-1$ & $0$ & $0$ & $0$ & $\mathbb{Z}$ & $0$ & $\mathbb{Z}_2$ & $\mathbb{Z}_2$ & $\mathbb{Z}$ & $0$ \\
$\text{CI}$ & $1$ & $-1$ & $1$ & $0$ & $0$ & $0$ & $\mathbb{Z}$ & $0$ & $\mathbb{Z}_2$ & $\mathbb{Z}_2$ & $\mathbb{Z}$ \\ \hline \hline

\end{tabular}
\caption{Periodic table of topological insulators and topological superconductors. Labels in the first column denote the Cartan classes of Hamiltonians, $T$, $C$ and $S$ denote the time-reversal, charge conjugation and chiral symmetries, respectively while 0 corresponds to the absence of the symmetry, +1 or -1 corresponds to the square of the symmetry operation. Numbers in the following columns denote the dimension and $\mathbb{Z}$ and $\mathbb{Z}_2$ groups correspond to the number of topological phases while 0 corresponds to the absence of topological phases.}
\end{table}

Since the electronic band structure of topological materials can be described in terms of vector bundles on the BZ, the classification of topological phases can be derived from K-theory groups of the point which classify the stably equivalent vector bundles \cite{Kitaev}. The first two rows in the periodic table correspond to complex classes and classified by $K^{-k}(\text{pt})$ groups which are
\[
K(\text{pt})=\mathbb{Z}\quad\quad,\quad\quad K^{-1}(\text{pt})=0
\]
and periodic with respect to mod 2. The remaining eight rows in the periodic table correspond to real classes and classified by the real K-theory groups $KO^{-k}(\text{pt})$ which are given as follows

\quad\\
{\centering{
\begin{tabular}{c c c c c c c c c}

\hline \hline
$k$ & $0$ & $1$ & $2$ & $3$ & $4$ & $5$ & $6$ & $7$ \\ \hline
$KO^{-k}(\text{pt})$ & $\mathbb{Z}$\, & $\mathbb{Z}_2$ & $\mathbb{Z}_2$ & $0\,\,$ & $\mathbb{Z}$\,\, & $0\,\,$ & $0\,\,$ & $0$ \\ \hline \hline

\end{tabular}}
\quad\\
\quad\\
\quad\\}

and periodic with respect to mod 8. So, the periodic table of topological insulators and superconductors is in one-to-one correspondence with K-theory groups given in the Tables II and III. The details of this construction can be found in \cite{Ertem}.
\begin{table}
\centering
\begin{tabular}{c c c}
\hline \hline
$K^{-(s-n) (\text{mod }2)}$ & $n=0$ & $1$ \\ \hline
$s=0$ & $\mathbb{Z}$ & $0$ \\
$1$ & $0$ & $\mathbb{Z}$ \\ \hline \hline
\end{tabular}
\caption{Complex K-theory goups at the point.}
\end{table}

\begin{table}
\centering
\begin{tabular}{c c c c c c c c c}

\hline \hline
$KO^{-(s-n) (\text{mod }8)}(\text{pt})$ & $n=0$ & $1$ & $2$ & $3$ & $4$ & $5$ & $6$ & $7$ \\ \hline
$s=0$ & $\mathbb{Z}$ & $0$ & $0$ & $0$ & $\mathbb{Z}$ & $0$ & $\mathbb{Z}_2$ & $\mathbb{Z}_2$ \\
$1$ & $\mathbb{Z}_2$ & $\mathbb{Z}$ & $0$ & $0$ & $0$ & $\mathbb{Z}$ & $0$ & $\mathbb{Z}_2$ \\
$2$ & $\mathbb{Z}_2$ & $\mathbb{Z}_2$ & $\mathbb{Z}$ & $0$ & $0$ & $0$ & $\mathbb{Z}$ & $0$ \\
$3$ & $0$ & $\mathbb{Z}_2$ & $\mathbb{Z}_2$ & $\mathbb{Z}$ & $0$ & $0$ & $0$ & $\mathbb{Z}$ \\
$4$ & $\mathbb{Z}$ & $0$ & $\mathbb{Z}_2$ & $\mathbb{Z}_2$ & $\mathbb{Z}$ & $0$ & $0$ & $0$ \\
$5$ & $0$ & $\mathbb{Z}$ & $0$ & $\mathbb{Z}_2$ & $\mathbb{Z}_2$ & $\mathbb{Z}$ & $0$ & $0$ \\
$6$ & $0$ & $0$ & $\mathbb{Z}$ & $0$ & $\mathbb{Z}_2$ & $\mathbb{Z}_2$ & $\mathbb{Z}$ & $0$ \\
$7$ & $0$ & $0$ & $0$ & $\mathbb{Z}$ & $0$ & $\mathbb{Z}_2$ & $\mathbb{Z}_2$ & $\mathbb{Z}$ \\ \hline \hline

\end{tabular}
\caption{Real K-theory groups at the point.}
\end{table}

The above classification of topological insulators and superconductors can also be characterized in terms of spin cobordism groups. Depending on the topological structure of manifolds, one can define equivalence relations between manifolds which give way to cobordism groups. For a manifold $M$, if the first Stiefel-Whitney class of the tangent bundle of $M$ vanishes $w_1(M)=0$, then $M$ is an oriented manifold. Two oriented $k$-manifolds $X$ and $Y$ are called equivalent if and only if they are cobordant, namely if there is an oriented $(k+1)$-manifold $W$ whose boundary consists of the union of $X$ and $Y$ that is $\partial W=\overline{X}\cup Y$ where $\overline{X}$ denotes the reversed orientation of $X$. This equivalence relation of manifolds determine the oriented cobordism group $\Omega^{SO}_k$ in dimension $k$. Moreover, if the second Stiefel-Whitney class of $M$ also vanishes $w_2(M)=0$, then $M$ corresponds to a spin manifold and the cobordism equivalence relation between spin $k$-manifolds determine the spin cobordism group $\Omega^{Spin}_k$ in dimension $k$.

\begin{proposition}
Real classes of topological insulators and superconductors can be classified by spin cobordism groups $\Omega^{Spin}_k$ for dimensions $\leq 7$.
\end{proposition}
\begin{proof}
The classification of topological insulators in terms of spin cobordism groups arises from the relation between $KO$-theory groups and spin cobordism groups. There is a ring homomorphism between $\Omega^{Spin}_k$ and $KO^{-k}(\text{pt})$ \cite{Lawson Michelsohn}
\[
\widehat{\mathcal{A}}_k:\Omega^{Spin}_k\longrightarrow KO^{-k}(\text{pt}).
\]
This homomorphism is an isomorphism for $k\leq 7$. Then, we have the spin cobordism groups $\Omega^{Spin}_k$ for $k\leq 7$ as follows

\quad\\
{\centering{
\begin{tabular}{c c c c c c c c c}

\hline \hline
$k$ & $0$ & $1$ & $2$ & $3$ & $4$ & $5$ & $6$ & $7$ \\ \hline
$\Omega^{Spin}_k$ & $\mathbb{Z}$\, & $\mathbb{Z}_2$ & $\mathbb{Z}_2$ & $0\,\,$ & $\mathbb{Z}$\,\, & $0\,\,$ & $0\,\,$ & $0$ \\ \hline \hline

\end{tabular}}
\quad\\
\quad\\
\quad\\}

For $k>7$, $\Omega^{Spin}_k$ groups differ from $KO^{-k}(\text{pt})$ groups. So, there is no eightfold periodicity for spin cobordism groups and the classification of topological insulators can only be possible in terms of spin cobordism groups for dimensions $n\leq 7$. The table of spin cobordism groups which is equivalent to the table of $KO^{-k}(\text{pt})$ groups can be written as in Table IV for $k\leq 7$. Since $KO^{-k}(\text{pt})$ groups classify only the real classes of topological insulators and superconductors, the spin cobordism groups classify the real classes of topological insulators and superconductors for dimensions $\leq 7$.
\end{proof}

\begin{table}
\centering
\begin{tabular}{c c c c c c c c c}

\hline \hline
$\Omega^{Spin}_{s-n(\text{mod 8})}$ & $n=0$ & $1$ & $2$ & $3$ & $4$ & $5$ & $6$ & $7$ \\ \hline
$s=0$ & $\mathbb{Z}$ & $0$ & $0$ & $0$ & $\mathbb{Z}$ & $0$ & $\mathbb{Z}_2$ & $\mathbb{Z}_2$ \\
$1$ & $\mathbb{Z}_2$ & $\mathbb{Z}$ & $0$ & $0$ & $0$ & $\mathbb{Z}$ & $0$ & $\mathbb{Z}_2$ \\
$2$ & $\mathbb{Z}_2$ & $\mathbb{Z}_2$ & $\mathbb{Z}$ & $0$ & $0$ & $0$ & $\mathbb{Z}$ & $0$ \\
$3$ & $0$ & $\mathbb{Z}_2$ & $\mathbb{Z}_2$ & $\mathbb{Z}$ & $0$ & $0$ & $0$ & $\mathbb{Z}$ \\
$4$ & $\mathbb{Z}$ & $0$ & $\mathbb{Z}_2$ & $\mathbb{Z}_2$ & $\mathbb{Z}$ & $0$ & $0$ & $0$ \\
$5$ & $0$ & $\mathbb{Z}$ & $0$ & $\mathbb{Z}_2$ & $\mathbb{Z}_2$ & $\mathbb{Z}$ & $0$ & $0$ \\
$6$ & $0$ & $0$ & $\mathbb{Z}$ & $0$ & $\mathbb{Z}_2$ & $\mathbb{Z}_2$ & $\mathbb{Z}$ & $0$ \\
$7$ & $0$ & $0$ & $0$ & $\mathbb{Z}$ & $0$ & $\mathbb{Z}_2$ & $\mathbb{Z}_2$ & $\mathbb{Z}$ \\ \hline \hline

\end{tabular}
\caption{Spin cobordism groups for $s\leq 7$ and $n\leq 7$.}
\end{table}

Low energy limit of band structures of topological materials are described by Dirac Hamiltonians and the eigenvalue equation of these Hamiltonians correspond to the Dirac equation of pseudo-fermions. So, the BZ of topological materials correspond to spin manifolds on which the Dirac operators can be defined. There can be different spin structures on spin manifolds. Since the choice of different spin structures determine the eigenvalue structure of the Dirac operators, the topological classes determined by the spin cobordism groups are related to the choices of different spin structures on the BZ. For example, on a $d$-dimensional torus $T^d$, there are $2^d$ different spin structures and the eigenvalues of the Dirac operator on $T^d$ depend on the choice of the spin structure \cite{Bar}. The number of different spin structures on a spin manifold $M$ is determined by the first cohomology group $H^1(M, \mathbb{Z}_2)$ and the spin cobordism group $\Omega^{Spin}_k$ correspond to equivalence of $k$-manifolds with the same spin structure.

For two topological materials with common boundary, if they are in the same topological class, then there is no topological discontinuity along the boundary and there is no edge current between them. However, if they are in different topological classes, then there is a topological discontinuity at the boundary which produce edge currents along the boundary. To describe this fact in terms of spin cobordism groups, let us consider two $(k+1)$-manifolds $M$ and $M'$ with a common $k$-dimensional boundary $N$ that is $\partial M=\partial M'=N$. If $H^1(N,\mathbb{Z}_2)=0$, then there is a unique spin structure on $N$ and different spin structures on $M$ and $M'$ reduce to the same spin structure on $N$. This means that if $M$ and $M'$ correspond to BZ of topological materials, then $M$ and $M'$ are in the same topological class since we have $\Omega^{Spin}_k=0$ because of the unique spin structure on $N$. For example, if we have a Dirac Hamiltonian $\mathcal{H}_M$ on $M$ with spin structure $\sigma_M$ and a Dirac Hamiltonian $\mathcal{H}_{M'}$ on $M'$ with spin structure $\sigma_{M'}$, then they both reduce to a Dirac Hamiltonian $\mathcal{H}_N$ on the boundary $N$ with the same spin structure $\sigma_N$ and hence they are in the same topological class. On the other hand, if we have $H^1(N,\mathbb{Z}_2)\neq 0$, then there are more than one spin structure on $N$ and if spin structures of $M$ and $M'$ reduce to different spin structures on $N$, then topological materials corresponding to $M$ and $M'$ are in different topological classes since we have $\Omega^{Spin}_k\neq 0$ in that case. For example, Dirac Hamiltonian $\mathcal{H}_M$ on $M$ with spin structure $\sigma_M$ will reduce to Dirac Hamiltonian $\mathcal{H}_{N_1}$ on the boundary $N$ with spin structure $\sigma_{N_1}$ and Dirac Hamiltonian $\mathcal{H}_{M'}$ on $M'$ with spin structure $\sigma_{M'}$ will reduce to $\mathcal{H}_{N_2}$ on the boundary $N$ with another spin structure $\sigma_{N_2}$ as shown in Fig. 1. Since the eigenvalue structures of $\mathcal{H}_M$ and $\mathcal{H}_{M'}$ are determined by the chosen spin structures and these choices reduce to the different spin structures on the common boundary they will be in different topological classes in this case.

\begin{figure}
  \includegraphics[width=\linewidth]{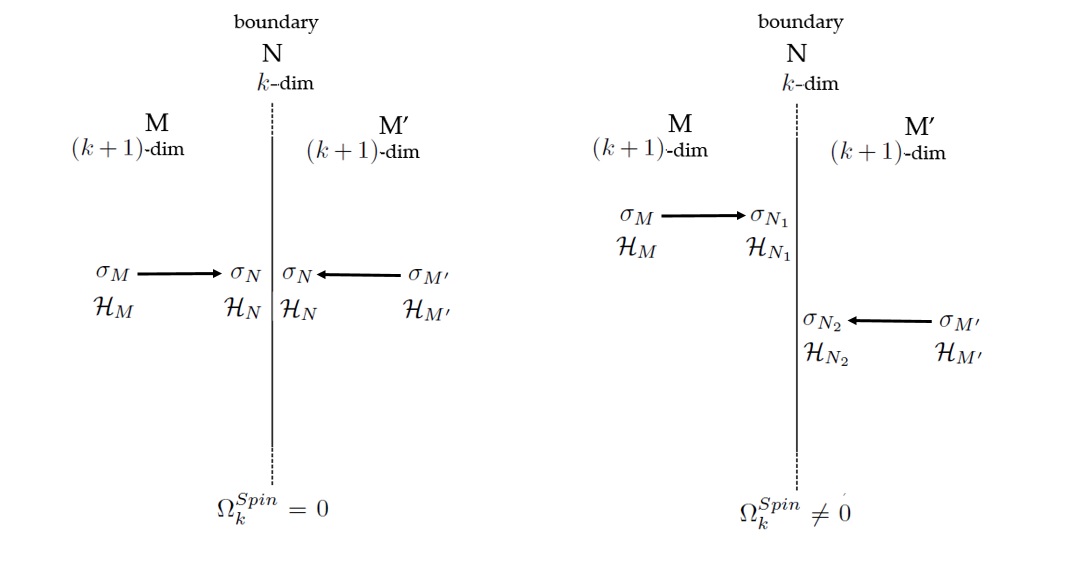}
  \caption{If the spin cobordism group vanishes at the boundary, then the Hamiltonians with different spin structures will reduce to the same spin structure and Hamiltonian at the boundary. If the spin cobordism group at the boundary is non-trivial, then the Hamiltonians with different spin structures for different topological classes will reduce to different spin structures and Hamiltonians at the boundary.}
  \label{fig:spin1}
\end{figure}

\subsection{Topological invariants}

Topological classes of topological insulators and superconductors are determined by indices of $Cl_k$-Dirac operators $\displaystyle{\not}D_k$ corresponding to Dirac Hamiltonians and those indices are valued in $KO^{-k}(\text{pt})$ groups \cite{Lawson Michelsohn}. These topological invariants are given for real classes as follows
\begin{eqnarray}
\text{ind}(\displaystyle{\not}D_k)&=&\left\{
                                                                               \begin{array}{ll}
                                                                                 \text{dim}_{\mathbb{C}}\textbf{H}_k (\text{mod }2)\in\mathbb{Z}_2, & \hbox{ for $k\equiv 1$  $(\text{mod }8)$} \\
                                                                                 \text{dim}_{\mathbb{H}}\textbf{H}_k (\text{mod }2)\in\mathbb{Z}_2, & \hbox{ for $k\equiv 2$  $(\text{mod }8)$} \\
                                                                                 \frac{1}{2}\widehat{A}(M)\in\mathbb{Z}, & \hbox{ for $k\equiv 4$  $(\text{mod }8)$} \\
                                                                                 \widehat{A}(M)\in\mathbb{Z}, & \hbox{ for $k\equiv 0$  $(\text{mod }8)$}
                                                                               \end{array}
                                                                             \right.
\end{eqnarray}
where $\textbf{H}_k$ correspond to the space of harmonic forms and dim$\textbf{H}_k(\text{mod }2)$ is also defined as the $\alpha$-genus of the manifold and $\widehat{A}(M)$ denotes the $\widehat{A}$-genus of the manifold $M$.

\begin{proposition}
The indices of Dirac operators, which are topological invariants of $\mathbb{Z}$ and $\mathbb{Z}_2$ groups appearing in the periodic table of topological insulators and superconductors, correspond to spin cobordism invariants of $\Omega^{Spin}_k$ which determine the number of different spin structures on the spin manifold.
\end{proposition}
\begin{proof}
For a Dirac Hamiltonian $\mathcal{H}(k)$, there is a corresponding Dirac operator $\displaystyle{\not}D$ and the index of $\displaystyle{\not}D$ depends on the choice of the spin structure. So, $\Omega^{Spin}_k$ invariants are equivalent to $\text{ind}(\displaystyle{\not}D_k)$ because of the ring homomorphism $\Omega^{Spin}_*\longrightarrow KO^{-*}(\text{pt})$. Moreover, there is also a surjective graded-ring homomorphism between spin structure preserving diffeomorphism classes of connected spin manifolds $\textbf{M}^{Spin}_k$ and $KO^{-k}(\text{pt})$ groups \cite{Lawson Michelsohn};
\[
\text{ind}_*:\textbf{M}^{Spin}_*\longrightarrow KO^{-*}(\text{pt}).
\]
So, the chosen spin structures are dependent on the $KO^{-*}(\text{pt})$ invariants and hence on the indices of Dirac operators $\displaystyle{\not}D_k$. For example, $\mathbb{Z}_2$ invariants of ind$(\displaystyle{\not}D_k)$  for $k=1$ and $k=2$ are equivalent to $\eta$ invariant of Dirac operators for $\Omega^{Spin}_1$ and the Arf invariant for $\Omega^{Spin}_2$, respectively. Moreover, Fu-Kane-Mele $\mathbb{Z}_2$ invariant of $\mathbb{Z}_2$ topological insulators correspond to the Arf invariant of $\Omega^{Spin}_2$. $\mathbb{Z}$ invariants $\widehat{A}(M)$ are related to the signature $\sigma(M)$ of the manifold for $\Omega^{Spin}_0$ and $\Omega^{Spin}_4$

\quad\\
{\centering{
\begin{tabular}{c c c c c c c c c}

\hline \hline
$k$ & $0$ & $1$ & $2$ & $3$ & $4$ & $5$ & $6$ & $7$ \\ \hline
$\Omega^{Spin}_k$ & $\sigma(M)$\, & $\eta\,\,$ & $\text{Arf}$ & $-\,\,$ & $\sigma(M)$\,\, & $-\,\,$ & $-\,\,$ & $-$ \\ \hline \hline

\end{tabular}}
\quad\\
\quad\\
\quad\\}
\end{proof}

The eigenvalue structure of Dirac Hamiltonians determine the topological classes of topological insulators. Since the eigenvalue structure and the index of the corresponding Dirac operator depend on the chosen spin structure, spin cobordism invariants are related to the eigenvalues of Dirac Hamiltonians. For example, for $\mathbb{Z}_2$ invariants, the sign of the mass term $m$ in the Dirac Hamiltonian changes when there are odd number of pairs of zero modes of $\displaystyle{\not}D$ and this depends on the chosen spin structure. Namely, for $\text{ind}(\displaystyle{\not}D)\neq 0$, the partition function has a different sign for $m>0$ and $m<0$. This is related to the sign-changing mass terms in different topological classes of topological materials. If a massive Dirac Hamiltonian has a sign-changing mass term depending on the choice of the spin structure, then this reflects the non-trivial spin cobordism group at the boundary.

\section{Weyl semimetals and Spin$^c$ cobordism groups}

Let us denote the $d$-dimensional BZ of a Weyl semimetal by $T$. The Bloch Hamiltonian $H({\textbf{k}})$ of the Weyl semimetal will have zero eigenvalues at distinct points $w_i$ of $T$. So, we define the set of Weyl points as $W=\bigcup_iw_i$ and $W\subset T$ corresponds to the gapless points of $H({\textbf{k}})$. The Bloch Hamiltonian $H({\textbf{k}})$ is gapped out of the Weyl points, namely on $T\backslash W$ and it corresponds to the topological insulator part of the Hamiltonian of the Weyl semimetal. Topological classes of the topological insulator part of the Hamiltonian are determined by the characteristic classes in the cohomology groups $H^*(T\backslash W)$ \cite{Mathai Thiang1, Mathai Thiang2}. We consider open $d$-balls $D_{w_i}$ surrounding each of the Weyl points $w_i$ and the disjoint union of these open $d$-balls corresponds to $D_W=\bigsqcup_{w_i\in W}D_{w_i}$. Boundaries of these $d$-balls are $(d-1)$-spheres $S_{w_i}$ around each Weyl point and the disjoint union of them is $S_W=\bigsqcup_{w_i\in W}S_{w_i}$. So, we have the following relations between $T$ and its subsets
\[
T=T\backslash W\cup D_W
\]
\[
T\backslash W\cap D_W=S_W.
\]
Topological classes of Weyl semimetals can be deduced from the cohomology groups of the subsets of BZ via Mayer-Vietoris principle as discussed in \cite{Mathai Thiang1, Mathai Thiang2}. For two manifolds $M_1$ and $M_2$, Mayer-Vietoris principle corresponds to the long exact sequence of cohomology groups
\begin{eqnarray}
{...}&\longrightarrow& H^{*-1}(M_1)\oplus H^{*-1}(M_2)\longrightarrow H^{*-1}(M_1\cap M_2)\longrightarrow H^{*-1}(M_1\cup M_2)\longrightarrow\nonumber\\
&\longrightarrow& H^*(M_1)\oplus H^*(M_2)\longrightarrow H^*(M_1\cap M_2)\longrightarrow H^*(M_1\cup M_2)\longrightarrow {...}\nonumber
\end{eqnarray}
If we define $M_1=T\backslash W$ and $M_2=D_W$, the we have
\begin{eqnarray}
{...}&\longrightarrow& H^{*-1}(T\backslash W)\oplus H^{*-1}(D_W)\longrightarrow H^{*-1}(S_W)\longrightarrow H^{*-1}(T)\longrightarrow\nonumber\\
&\longrightarrow& H^*(T\backslash W)\oplus H^*(D_W)\longrightarrow H^*(S_W)\longrightarrow H^*(T)\longrightarrow {...}\nonumber
\end{eqnarray}
Here, the cohomology groups in the long exact sequence correspond to the following topological invariants
\begin{eqnarray}
H^{*-1}(T)&;& \text{topological insulator in one lower dimension}\nonumber\\
H^*(T\backslash W)&;& \text{topological insulator part (gapped part) of Weyl semimetal}\nonumber\\
H^*(S_W)&;& \text{local Weyl charges at Weyl points}\nonumber\\
H^*(T)&;& \text{total Weyl charge (sum of local charges)}\nonumber
\end{eqnarray}
These cohomological invariants correspond to the topological invariants of topological materials since they can be generalized to K-theory invariants because of the fact that K-theory is a generalized cohomology theory. Indeed, index theorems show that the characteristic classes corresponding to K-theory invariants take values in cohomology groups. As an exposition of this fact, we also have a Mayer-Vietoris sequence for K-theory groups \cite{Karoubi}
\begin{eqnarray}
{...}&\longrightarrow& K^{-*-1}(M_1)\oplus K^{-*-1}(M_2)\longrightarrow K^{-*-1}(M_1\cap M_2)\longrightarrow K^{-*-1}(M_1\cup M_2)\longrightarrow\nonumber\\
&\longrightarrow& K^{-*}(M_1)\oplus K^{-*}(M_2)\longrightarrow K^{-*}(M_1\cap M_2)\longrightarrow K^{-*}(M_1\cup M_2)\longrightarrow {...}\nonumber
\end{eqnarray}

This analysis shows that if we cut the open balls around the Weyl points of Weyl semimetal BZ, we obtain a manifold corresponding to a gapped topological insulator BZ. Then, Weyl semimetals are extensions of topological insulators by adding singular Weyl points and open balls around them to the BZ of topological insulators. We will show below that this extension is equivalent to the extension of the spin structures on topological insulator BZ to spin$^c$ structures on Weyl semimetal BZ. This means that the spin cobordism classification of topological insulators can be extended to spin$^c$ cobordism classification for Weyl semimetals.

Definition of the spin structures on a manifold $M$ can be generalized to the definition of spin$^c$ structures. When the second Stiefel-Whitney class of $M$ does not vanish $w_2(M)\neq0$, then a spin structure cannot be defined on $M$. However, if $w_2(M)$ can be lifted to an integral cohomology class $\underline{w}$, then a spin$^c$ structure can be defined and spin$^c$ structures are labelled by the cohomology group $H^2(M;\mathbb{Z})$. The integral lift $\underline{w}$ can be represented by a dual embedded characteristic surface $\Sigma$ in $M$. The complement of $\Sigma$ that is $M\backslash\Sigma$ always has a spin structure, but this structure cannot be extended across $\Sigma$. Spin$^c$ structures correspond to generalizations of spin structures to extend them across $\Sigma$ and define across all $M$ \cite{Scorpan, Sati}.

Spin$^c$ cobordism groups $\Omega^{Spin^c}_k$ can also be defined as in the case of spin cobordism groups. Two spin$^c$ $k$-manifolds $M$ and $M'$ are equivalent if there is a $(k+1)$-manifold $W$ such that $\partial W=\overline{M}\cup M'$ and $W$ has a spin$^c$ structure restricted to the spin$^c$ structures on $M$ and $M'$. $\Omega^{Spin^c}_k$ can also be described in terms of characteristic surfaces \cite{Scorpan}. Let us consider a $k$-manifold $M$ with a characteristic surface $\Sigma$ and a chosen spin structure in the complement of $\Sigma$. Similarly, we have another pair $M'$ and $\Sigma'$ with a chosen spin structure in the complement of $\Sigma'$. The spin structures in the complements of characteristic surfaces can be extended as spin$^c$ structures in all of $M$ and $M'$. A spin$^c$ cobordism equivalence between $M$ and $M'$ can be defined as follows. If there is a $(k+1)$-manifold $W$ with a chosen 3-submanifold $Y$ such that $\partial W=\overline{M}\cup M'$ and $\partial Y=\overline{\Sigma}\cup\Sigma'$ and the spin structure on $W\backslash Y$ reduces to the chosen spin structures on $M\backslash\Sigma$ and $M'\backslash\Sigma'$, then $M$ and $M'$ are called spin$^c$ cobordant. Spin$^c$ cobordism groups for $k\leq 7$ are given as in the following table

\quad\\
{\centering{
\begin{tabular}{c c c c c c c c c}

\hline \hline
$k$ & $0$ & $1$ & $2$ & $3$ & $4$ & $5$ & $6$ & $7$ \\ \hline
$\Omega^{Spin^c}_k$ & $\,\,\,\mathbb{Z}$\,\,\,\,\, & $0\,\,\,$ & $\mathbb{Z}$\,\,\, & $\,\,\,0\,\,\,$ & $\,\,\,\mathbb{Z}\oplus\mathbb{Z}$\,\,\, & $0\,\,\,$ & $\mathbb{Z}\oplus\mathbb{Z}$\,\,\, & $0$ \\ \hline \hline

\end{tabular}}
\quad\\
\quad\\
\quad\\}

We can use the extension of spin structures to spin$^c$ structures for classifying Weyl semimetals. Topological insulators can be classified in terms of the possible choices of spin structures in BZ as discussed in Section II. A similar classification can be described for Weyl semimetals in terms of spin$^c$ cobordism groups. 

\begin{proposition}
Weyl semimetals can be classified by spin$^c$ cobordism groups $\Omega^{Spin^c}_k$ for dimensions $\leq 7$ via the extension of spin structures to spin$^c$ structures.
\end{proposition}
\begin{proof}
Let us denote the Weyl points in Weyl semimetals as $W$ which correspond to stable points in the band structure against perturbations and prevent the system being an insulator. The band structure of Weyl semimetals  whose BZ is denoted by $T$ are equivalent to the topological insulator band structures out of the open balls $D_W$ around Weyl points. So, BZ of a Weyl semimetal without Weyl points $T\backslash W$ has a spin structure, but this spin structure cannot be extended to the Weyl points. However, the spheres $S_W$ around Weyl points correspond to characteristic surfaces and spin$^c$ structures can be defined in all BZ including Weyl points via the extension procedure of spin structures to spin$^c$ structures through characteristic surfaces discussed above. Hence, we can obtain extendable spin$^c$ structures in Weyl semimetal BZ $T$ via spheres $S_W$ around Weyl points corresponding to characteristic surfaces from nonextendable spin structures in topological insulator BZ $T\backslash W$. As a result, the extensions of spin structures to spin$^c$ structures correspond to the extensions of topological insulators to topological Weyl semimetals. Then, topological Weyl semimetal classes are classified by spin$^c$ cobordism classes from the correspondence in Mayer-Vietoris sequences. Namely, the choices of different spin$^c$ structures in BZ correspond to different topological classes of Weyl semimetals. Similar to the case of spin Dirac operator $\displaystyle{\not}D$ and its relation with the choice of spin structures, the eigenvalue structure and index of the spin$^c$ Dirac operator $\displaystyle{\not}D^A$ coupled to the $U(1)$ potential $A$ is also dependent on the choice of spin$^c$ structure in BZ. So, the choice of the spin$^c$ structure determine the eigenvalue structure of $\displaystyle{\not}D_A$ and a change in the spin$^c$ structure leads to a change in the Weyl semimetal topological class. Weyl semimetals which are extensions of topological insulators in corresponding real symmetry classes can be classified by spin$^c$ cobordism groups as in Table V for dimensions $n\leq 7$. This table corresponds to the classification of Weyl semimetals for different classes and dimensions. The rows in the table denoted by the label $s$ corresponds to the real symmetry classes of the topological insulator parts of Weyl semimetals and not to the symmetry classes of Weyl semimetals themselves.
\end{proof}

\begin{table}
\centering
\begin{tabular}{c c c c c c c c c}

\hline \hline
$\Omega^{Spin^c}_{s-n(\text{mod 8})}$ & $n=0$ & $1$ & $2$ & $3$ & $4$ & $5$ & $6$ & $7$ \\ \hline
$s=0$ & $\mathbb{Z}$ & $0$ & $\mathbb{Z}\oplus\mathbb{Z}$ & $0$ & $\mathbb{Z}\oplus\mathbb{Z}$ & $0$ & $\mathbb{Z}$ & $0$ \\
$1$ & $0$ & $\mathbb{Z}$ & $0$ & $\mathbb{Z}\oplus\mathbb{Z}$ & $0$ & $\mathbb{Z}\oplus\mathbb{Z}$ & $0$ & $\mathbb{Z}$ \\
$2$ & $\mathbb{Z}$ & $0$ & $\mathbb{Z}$ & $0$ & $\mathbb{Z}\oplus\mathbb{Z}$ & $0$ & $\mathbb{Z}\oplus\mathbb{Z}$ & $0$ \\
$3$ & $0$ & $\mathbb{Z}$ & $0$ & $\mathbb{Z}$ & $0$ & $\mathbb{Z}\oplus\mathbb{Z}$ & $0$ & $\mathbb{Z}\oplus\mathbb{Z}$ \\
$4$ & $\mathbb{Z}\oplus\mathbb{Z}$ & $0$ & $\mathbb{Z}$ & $0$ & $\mathbb{Z}$ & $0$ & $\mathbb{Z}\oplus\mathbb{Z}$ & $0$ \\
$5$ & $0$ & $\mathbb{Z}\oplus\mathbb{Z}$ & $0$ & $\mathbb{Z}$ & $0$ & $\mathbb{Z}$ & $0$ & $\mathbb{Z}\oplus\mathbb{Z}$ \\
$6$ & $\mathbb{Z}\oplus\mathbb{Z}$ & $0$ & $\mathbb{Z}\oplus\mathbb{Z}$ & $0$ & $\mathbb{Z}$ & $0$ & $\mathbb{Z}$ & $0$ \\
$7$ & $0$ & $\mathbb{Z}\oplus\mathbb{Z}$ & $0$ & $\mathbb{Z}\oplus\mathbb{Z}$ & $0$ & $\mathbb{Z}$ & $0$ & $\mathbb{Z}$ \\ \hline \hline

\end{tabular}
\caption{Spin$^c$ cobordism groups for $s\leq 7$ and $n\leq 7$.}
\end{table}

The discontinuity at the common boundary of Weyl semimetals in different topological classes can be described in terms of spin$^c$ structures in a similar way to the topological insulator case. Let us consider two $(k+1)$-manifolds $M$ and $M'$ with a common $k$-dimensional boundary $N$ that is $\partial M=\partial M'=N$. If $H^2(N,\mathbb{Z})=0$, then there is a unique spin$^c$ structure on $N$ and different spin$^c$ structures on $M$ and $M'$ reduce to the same spin$^c$ structure on $N$. This means that if $M$ and $M'$ correspond to BZ of Weyl semimetals, then $M$ and $M'$ are in the same topological class since we have $\Omega^{Spin^c}_k=0$ because of the unique spin$^c$ structure on $N$. For example, if we have a Dirac Hamiltonian $\mathcal{H}_M$ on $M$ with spin$^c$ structure $\sigma_M$ and a Dirac Hamiltonian $\mathcal{H}_{M'}$ on $M'$ with spin$^c$ structure $\sigma_{M'}$, then they both reduce to a Dirac Hamiltonian $\mathcal{H}_N$ on the boundary $N$ with the same spin$^c$ structure $\sigma_N$ and hence they are in the same topological class. On the other hand, if we have $H^2(N,\mathbb{Z})\neq 0$, then there are more than one spin$^c$ structure on $N$ and if spin$^c$ structures of $M$ and $M'$ reduce to different spin$^c$ structures on $N$, then Weyl semimetals corresponding to $M$ and $M'$ are in different topological classes since we have $\Omega^{Spin^c}_k\neq 0$ in that case. For example, Dirac Hamiltonian $\mathcal{H}_M$ on $M$ with spin$^c$ structure $\sigma_M$ will reduce to Dirac Hamiltonian $\mathcal{H}_{N_1}$ on the boundary $N$ with spin$^c$ structure $\sigma_{N_1}$ and Dirac Hamiltonian $\mathcal{H}_{M'}$ on $M'$ with spin$^c$ structure $\sigma_{M'}$ will reduce to $\mathcal{H}_{N_2}$ on the boundary $N$ with another spin$^c$ structure $\sigma_{N_2}$ as shown in Fig. 2. Since the eigenvalue structures of $\mathcal{H}_M$ and $\mathcal{H}_{M'}$ are determined by the chosen spin$^c$ structures and these choices reduce to the different spin$^c$ structures on the common boundary they will be in different topological classes in this case.

\begin{figure}
  \includegraphics[width=\linewidth]{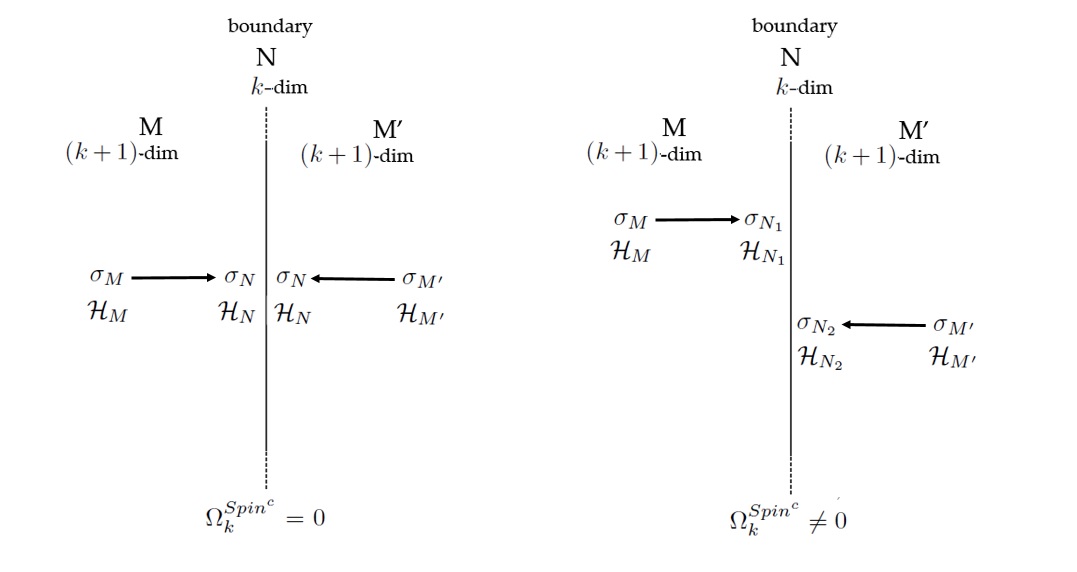}
  \caption{If the spin$^c$ cobordism group vanishes at the boundary, then the Hamiltonians with different spin$^c$ structures will reduce to the same spin$^c$ structure and Hamiltonian at the boundary. If the spin$^c$ cobordism group at the boundary is non-trivial, then the Hamiltonians with different spin$^c$ structures for different topological classes will reduce to different spin$^c$ structures and Hamiltonians at the boundary.}
  \label{fig:spinc}
\end{figure}

\subsection{Topological invariants}

Weyl points in Weyl semimetals are characterized by monopole charges corresponding to the Chern numbers defined from the Berry curvature on the spheres around Weyl points. Similarly, spin$^c$ structures are also characterized by monopole equations which are SW equations. So, the extension of topological insulators to Weyl semimetals via monopole charges on Weyl points can be described in terms of the extension of spin structures to spin$^c$ structures via the definition of SW invariants around singular points.

Let us consider a Weyl semimetal characterized by $\Omega^{Spin^c}_4$ with Brillouin torus $T$ and the collection of Weyl points $W$. Characteristic surfaces corresponding to the union of spheres around Weyl points are denoted by $S_W$. Topological invariants characterizing the Weyl semimetal arise from the spin$^c$ cobordism group $\Omega^{Spin^c}_4=\mathbb{Z}\oplus\mathbb{Z}$. The first $\mathbb{Z}$ corresponds to the topological insulator part invariant of $T\backslash W$ which takes values in $H^d(T\backslash W)$ and determined by the index of the Dirac Hamiltonian. The second $\mathbb{Z}$ corresponds to the Weyl point charge which takes values in $H^d(S_W)$ and determined by the Chern numbers around $S_W$.

Choice of spin$^c$ structures on a manifold is determined by SW invariants which are related to the index of the gauged Dirac operator and the monopole charge of the gauge curvature. SW invariants arise from SW equations which are written in 4-dimensions as follows \cite{Seiberg Witten, Friedrich}
\begin{eqnarray}
\displaystyle{\not}D^A\psi&=&0\nonumber\\
F^+_A&=&(\psi\overline{\psi})_2
\end{eqnarray}
where $\displaystyle{\not}D^A$ is the gauged Dirac operator coupled to $U(1)$ connection $A$ and $\psi$ is a section of the spin$^c$ bundle. $F^+_A$ is the self-dual part of the gauge curvature 2-form $F_A=dA$ and $(\psi\overline{\psi})_2$ is the bilinear 2-form constructed out of the spinor $\psi$. The number of solutions of SW equations determine the different spin$^c$ structures on the manifold.
\begin{proposition}
Topological invariants determined from Weyl semimetal equations are equivalent to the invariants resulting from SW equations.
\end{proposition}
\begin{proof}
We have the following equations for Weyl semimetals
\begin{eqnarray}
\mathcal{H}^A|\Psi>&=&E|\Psi>\nonumber\\
F_{ij}&=&<\Psi|\nabla_{[i}\nabla_{j]}\Psi>
\end{eqnarray}
where $\mathcal{H}^A$ is the Dirac Hamiltonian of the Weyl semimetal with gauge coupling, $|\Psi>$ is the eigenstate and $E$ is the eigenvalue of the Hamiltonian. $F_{ij}$ is the Berry curvature defined from the Berry connection $A_i=-i<\Psi|\nabla_i\Psi>$ as $F_{ij}=\nabla_iA_j-\nabla_jA_i$ where $i,j$ label the components. These equations correspond to gauged harmonic spinor equation and curvature equality of SW equations and solutions of these equations determine the topological invariants of Weyl semimetals. For the spin$^c$ structures in the same cobordism class $\Omega^{Spin^c}_d$, the eigenvalue structure of $\displaystyle{\not}D^A$ does not change and we stay in the same topological class of a Weyl semimetal. The above identification between SW equations and Weyl semimetal equations is relevant in 3 space and 1 time dimensions in the corresponding symmetry class for which Weyl points are well-defined and we have $\Omega^{Spin^c}_4=\mathbb{Z}\oplus\mathbb{Z}$. For other dimensions, generalizations of SW equations can be considered corresponding to the Dirac Hamiltonians and Berry curvatures in higher dimensions for $A_i=-i<\Psi|\nabla_i\Psi>$ and $F_{ij}=-2i<\Psi|\nabla_{[i}\nabla_{j]}\Psi>$. So, we have a correspondence between spin$^c$ structures and Weyl semimetal classes via Dirac operators and gauge connections;
\begin{eqnarray}
\text{spin$^c$ structure}\longrightarrow \displaystyle{\not}D\,\,,A \longrightarrow \text{Weyl semimetal}\nonumber
\end{eqnarray}
To see the equivalence between self-dual curvature $F^+_A$ in SW equations and the Berry curvature $F_{ij}$, one can check the correspondence between spin$^c$ inner products and Hermitian inner products in Hilbert space of quantum states. For the right hand side of the curvature equality of SW equations, spin$^c$ inner product $(\,,\,)$ can be used to construct the bilinear 2-form $(\psi\overline{\psi})_2$ as follows
\begin{eqnarray}
(\psi\overline{\psi})_2&=&(\psi, e_b.e_a.\psi)e^a\wedge e^b\nonumber\\
&=&\pm(e_b.\psi, e_a.\psi)e^a\wedge e^b
\end{eqnarray}
where $e^a$ are coframe basis, $.$ denotes Clifford multiplication and $\wedge$ is the wedge product of forms. On the other hand, the right hand side of the Berry curvature equality can be written as
\begin{equation}
<\Psi|\nabla_{[i}\nabla_{j]}\Psi>=<\nabla_{[i}\Psi|\nabla_{j]}\Psi>.
\end{equation}
Then, we have the following equivalence between coframe basis and partial derivatives in terms of momentum coordinates; $e^a\rightarrow i\nabla_{a}$. So, the Berry curvature is equivalent to the self-dual curvature condition of SW equations written in terms of momentum space coordinates for the quantum Hilbert space. Chern numbers defined from the integrals of self-dual curvature and the Berry curvature are also equal to each other. Similarly, tha gauged Dirac operator $\displaystyle{\not}D^A=e^a.(\nabla_{X_a}+i_{X_a}A)$ and the Dirac Hamiltonian with gauge coupling $\mathcal{H}^A(k_i)=i\nabla_i+A_i$ are equivalent to each other. 
\end{proof}

Spin$^c$ cobordism invariants for the presence of non-trivial classes in 0, 2, 4 and 6 dimensions can be described as follows \cite{Wan Wang}. In 0-dimensions, we have $\Omega^{Spin^c}_0=\mathbb{Z}$ and cobordism invariant correspond to the first Chern number $c_1$. In 2-dimensions, we have $\Omega^{Spin^c}_2=\mathbb{Z}$ and the cobordism invariant corresponds to $c_1/2$. Here $c_1$ is divided by 2 since $c_1\text{ mod }2=w_2(TM)$ and the second Stiefel-Whitney class vanishes on spin$^c$ 2-manifolds. In 4-dimensions, $\Omega^{Spin^c}_4=\mathbb{Z}\oplus\mathbb{Z}$ and the corresponding cobordism invariant for the first component is $\frac{\sigma-\Sigma.\Sigma}{8}$ where $\sigma$ is the signature invariant of the spin$^c$ 4-manifold $M$ and $\Sigma$ is a characteristic surface of it with $.$ denoting the intersection of surfaces. From the Rokhlin's theorem, $\sigma-\Sigma.\Sigma$ is a multiple of 8 and the invariant corresponds to the index of spin$^c$ Dirac operator $\displaystyle{\not}D^A$
\begin{eqnarray}
\text{ind }\displaystyle{\not}D^A=\frac{1}{8}(\sigma-\Sigma.\Sigma).
\end{eqnarray}
The invariant for the second component correspond to the Chern number $c_1^2$ along the characteristic surface which is the monopole charge. Then, spin$^c$ cobordism invariants correspond to Weyl semimetal topological invariants in $\Omega^{Spin^c}_4$ case which contain the topological insulator part $\text{ind}\displaystyle{\not}D^A$ and the monopole charge. In 6-dimensions, we have $\Omega^{Spin^c}_6=\mathbb{Z}\oplus\mathbb{Z}$ and the corresponding cobordism invariants are $c_1^3/2$ and $c_1\frac{\sigma}{16}$.

\quad\\
{\centering{
\begin{tabular}{c c c c c c c c c}

\hline \hline
$k$ & $0$ & $1$ & $2$ & $3$ & $4$ & $5$ & $6$ & $7$ \\ \hline
$\Omega^{Spin^c}_k$ & $c_1$\, & $\,\,-\,\,$ & $\frac{c_1}{2}$ & $\,\,-\,\,$ & $\frac{\sigma-\Sigma.\Sigma}{8}\oplus c_1^2$\,\, & $\,\,-\,\,$ & $\frac{c_1^3}{2}\oplus c_1\frac{\sigma}{16}\,\,$ & $-$ \\ \hline \hline

\end{tabular}}
\quad\\
\quad\\
\quad\\}

So, these invariants label the Weyl semimetal classes in various symmetry classes and dimensions and can be calculated from Weyl semimetal Hamiltonians and Berry curvatures by the equivalence of (10) and (11).

\subsection{Fermi arcs and spin$^c$ structures}

Fermi arcs are gapless surface modes of Weyl semimetals which connect the projections of the Weyl points with charges of opposite sign at the boundary of the semimetal. They are experimental signatures of Weyl semimetals.
\begin{proposition}
Fermi arcs at the boundaries of Weyl semimetals depend on the choice of spin$^c$ structures on the BZ of Weyl semimetals.
\end{proposition}
\begin{proof}
In mathematical terms, Fermi arcs at the boundaries of Weyl semimetals correspond to Euler structures of vector fields \cite{Mathai Thiang1} and Euler structures are classified by torsion invariants \cite{Turaev}. Moreover, there is a bijection between Euler structures and spin$^c$ structures \cite{Salamon}. This means that SW invariants and torsion invariants are equal to each other and both Euler structures and spin$^c$ structures can be naturally identified with the second cohomology group $H^2(M;\mathbb{Z})$ of the manifold $M$. So, the Fermi arcs at the boundaries of Weyl semimetals depend on the choice of the spin$^c$ structure. This corresponds to the spin$^c$ cobordism classification of Weyl semimetals via the classification of Fermi arcs at the boundary. Namely, the Fermi arcs of Weyl semimetals are results of spin$^c$ cobordism classification of Weyl semimetals.
\end{proof}

\section{Conclusion}

Classification of topological insulators and superconductors in different symmetry classes and dimensions can be managed in various ways. These approaches include the method of Clifford algebras and K-theory, analysing the stability of gapless boundary states against perturbations, dimensional reduction method of massive Dirac Hamiltonians and the method of Clifford chessboard and indices of Dirac operators. In this paper, we extend these classification approaches to spin cobordism groups via the homomorphism between K-theory groups and spin cobordism groups in dimensions less than or equal to seven. In that way, we interpret the classes of topological insulators in terms of possible spin structure choices in the BZ of materials. We also extend this discussion to the case of Weyl semimetals. By using the equivalence between the extension of spin structures to spin$^c$ structures and the extension of topological insulators to Weyl semimetals, we propose a classification scheme for Weyl semimetals in terms of spin$^c$ cobordism groups. The equivalence of topological invariants for Weyl semimetals and spin$^c$ cobordism groups and the relation between Fermi arcs and the choices of spin$^c$ structures also support this approach.

Since there are other types of topological semimetals such as Dirac semimetals and topological nodal-line semimetals, the extension of this approach to other types of materials can be investigated in a future work. This approach can be extended to a more general framework in terms of the relations between symmetry protected topological phases and cobordisms \cite{Wan Wang, Kapustin, Kapustin Thorngren Turzillo Wang, Yonekura} The generalization of spin and spin$^c$ groups to pin, pin$^c$ and Clifford groups can be used in this more general investigation. Moreover, this work can also lead to a complete framework for more general topological materials in terms of spin geometry.

%\references%

\end{document}